\documentclass[twoside,a4paper]{article} 
\usepackage{amsmath,graphicx,amssymb,fancyhdr,amsthm,enumerate,textcomp} 
\pdfoutput=1 
\usepackage[usenames]{color} 
\usepackage[latin1]{inputenc} 
\newtheorem{thm}{Theorem}[section]  
\newtheorem{cor}[thm]{Corollary}  
  
\newtheorem{prop}[thm]{Proposition}  
\newtheorem{con}[thm]{Conjecture}  
\theoremstyle{definition}  
\newtheorem{defn}[thm]{Definition}   
\theoremstyle{remark}   
   
\def\beq{\begin{eqnarray}}   
\def\eeq{\end{eqnarray}}   
\def\bsp{\begin{split}}   
\def\esp{\end{split}}

\def\Tr{\mathrm{Tr}}   
\def\d{\mathrm{d}}   
\def\diag{\mathrm{diag}}   
  
\def\T{ {\sf T} }

\def\redpoint{\diamond}  
\def\redtree{\overset{{\displaystyle\diamond}}{\bigwedge}}  
\def\tree{\bigwedge}  
   
\newcommand{\mf}[1]{{\mathfrak #1}}   
   
\newcommand{\mb}[1]{{\mathbb #1}}   
    
\newcommand{\mbold}[1]{\mbox{\boldmath{\ensuremath{#1}}}}

\begin{document}    
    
\title{\Large\textbf{Curvature  operators and scalar curvature invariants}}    
\author{{\large\textbf{Sigbj\o rn Hervik$^\text{\tiny\textleaf}$ and Alan Coley$^\heartsuit$} }    
%EndAName    
%\address{    
 \vspace{0.3cm} \\    
$^\text{\tiny\textleaf}$Faculty of Science and Technology,\\    
 University of Stavanger,\\  N-4036 Stavanger, Norway     
\vspace{0.2cm} \\ 
$^\heartsuit$Department of Mathematics and Statistics,\\ 
Dalhousie University, \\ 
Halifax, Nova Scotia, Canada B3H 3J5 
%\email{    
\vspace{0.3cm} \\     
\texttt{sigbjorn.hervik@uis.no, aac@mathstat.dal.ca} }    
\date{\today}    
\maketitle   
\pagestyle{fancy}    
\fancyhead{} % clear all header fields    
\fancyhead[EC]{S. Hervik \& A. Coley}    
\fancyhead[EL,OR]{\thepage}    
\fancyhead[OC]{Curvature Operators}    
\fancyfoot{} % clear all footer fields    
    
\begin{abstract}

We continue the study of the question of when a pseudo-Riemannain manifold can be locally 
characterised by its scalar polynomial curvature invariants 
(constructed from the Riemann tensor and its covariant derivatives).
We make further use of alignment theory and the bivector form of the Weyl operator 
in higher dimensions, and introduce the important notions of diagonalisability
and (complex) analytic metric extension. We show that if there exists an
analytic metric extension of an arbitrary dimensional space of any signature
to a Riemannian space (of Euclidean signature), then that
space is characterised by its scalar curvature invariants.  
In particular, we discuss the Lorentzian case and the neutral 
signature case in four dimensions in more detail.

\end{abstract}

\newpage
\section{Introduction}

Recently, the question of when a pseudo-Riemannain manifold can be locally 
characterized by its scalar polynomial curvature invariants 
constructed from the Riemann tensor and its covariant derivatives has been addressed.

In  \cite{inv} it was shown that in four dimensions (4d) a Lorentzian spacetime 
metric is either $\mathcal{I}$-non-degenerate, and hence locally 
characterized by its scalar polynomial curvature invariants,
or is a degenerate Kundt spacetime \cite{Kundt}. 
Therefore, the degenerate Kundt spacetimes are the only spacetimes 
in 4d that are {\em not} $\mathcal{I}$-non-degenerate, and their 
metrics are the only metrics not uniquely determined by their curvature 
invariants. 
In the proof of the $\mathcal{I}$-non-degenerate theorem in 
\cite{inv} it was necessary, for example,  to determine for which Segre types 
for the Ricci tensor the spacetime is 
$\mathcal{I}$-non-degenerate. By analogy, in higher dimensions it is useful to
utilize the bivector 
formalism for the Weyl tensor. 
Indeed, by defining the Weyl bivector operator 
in higher dimensions \cite{bivect}
and making use of the alignment theory \cite{class}, 
it is possible to algebraically classify any tensor (including the 
Weyl tensor utilizing the eigenbivector problem) in a Lorentzian 
spacetime of arbitrary dimensions, which has proven very useful
in contemporary theoretical physics.

 The link between a metric (or rather, when two metrics can be distinguished up to diffeomorphisms) 
 and the curvature tensors is provided through \emph{Cartan's equivalence principle}.
 An exact statement of this is given in \cite{exsol}; however, the idea is that knowing the 
 Riemann curvature tensor, and its covariant derivatives, with respect to a fixed 
  frame 
 determines the metric up to isometry. It should be pointed out that it may be difficult to actually \emph{construct} the metric, however, the metric is determined in principle. More precisely, given a point, $p$, a frame $E_\alpha$, and the corresponding Riemann tensor $R$ and its covariant derivatives $\nabla R, \dots, \nabla^{(k)}R,\dots $, then if there exists a map such that: 
\beq  
(p,E_\alpha)&\mapsto & (\bar{p},\bar{E}_\alpha), \quad \text{and} \nonumber  \\ 
 (R_{\alpha\beta\gamma\delta},R_{\alpha\beta\gamma\delta;\mu_1},...,R_{\alpha\beta\gamma\delta;\mu_1...\mu_k}) & \mapsto &  (\bar{R}_{\alpha\beta\gamma\delta},\bar{R}_{\alpha\beta\gamma\delta;\mu_1},...,\bar{R}_{\alpha\beta\gamma\delta;\mu_1...\mu_k}), \nonumber 
\eeq 
then there exists an isometry $\phi$ such that $\phi(p)=\bar{p}$ (which induces the abovementioned 
map). Therefore, if we know the components of the Riemann tensor and its covariant derivatives 
with respect to a fixed frame, then the geometry and, in principle, the metric is determined. 
The equivalence principle consequently manifests the connection between the curvature tensors 
and the spacetime metric. The question of whether we can, at least in principle, reconstruct the metric 
from the invariants thus hinges on the question whether we can reconstruct the curvature  
tensors from its scalar polynomial curvature invariants, indicated with a question mark in the following figure:
\[ \begin{array}{cl} \boxed{\text{Metric}, g_{\mu\nu}} \\
{\Updownarrow} & \hspace{-1cm} \leftarrow\text{Equivalence principle} \\
 \boxed{\text{Curvature tensors}}  \\ 
\Downarrow \quad \Uparrow ? \\
\boxed{\text{Invariants}}
\end{array}
\]
It is to address this question, the curvature operators are particularly useful.

In this paper we continue the study of pseudo-Riemannain manifolds and their local 
scalar polynomial curvature invariants. In principle, we
are interested
in the case of arbitrary dimensions and all signatures, although we are primarily  interested
from a physical point of view in the Lorentzian case and most illustrations will be done in
4d. We also briefly
discuss the neutral 
signature case.

First, we shall  discuss the concept of a differentiable manifold being characterised by 
its scalar invariants in more detail.
Suppose we consider the set of invariants as a function of the metric  
and its derivatives: $\mf{I}:g_{\mu\nu}\mapsto \mathcal{I}$ 
(the $\mf{I}$-map);
we are then interested in under what circumstances, given a set 
of invariants $\mathcal{I}_0$, 
this function has an inverse 
$\mf{I}^{-1}(\mathcal{I}_0)$. 
For a metric which is  
$\mathcal{I}$-non-degenerate the invariants characterize the  
spacetime uniquely, at least locally, in the space of (Lorentzian)  
metrics, which means that we can therefore distinguish such metrics  
using their curvature invariants and, hence,  $\mf{I}^{-1}(\mathcal{I}_0)$ is unique
(up to diffeomorphisms).

We then further discuss curvature operators (and introduce an operator calculus)
and their relationship to spacetime invariants.  
Any even-ranked tensor can be
considered as an operator, and there is a natural matrix representation of the
operator ${\sf T}$. 
In particular, for a curvature operator
(such as the zeroth order curvature operator, $\mathcal{R}_0$)
we can consider an eigenvector ${\sf   
v}$ with eigenvalue $\lambda$; i.e., ${\sf T}{\sf v}=\lambda{\sf   
v}$. This allows us to introduce the important notion of {\emph {diagonalisability}}.
In addition, we then 
introduce the notion of {\emph  {analytic metric extension}},
which is a complex analytic continuation  of a
real metric
under more general coordinate transformations than the real diffeomorphisms
which results in a  real bilinear form but which does not necessarily
preserve the metric 
signature.  
We shall show that if a space, $(\mathcal{M},g_{\mu\nu})$, of any signature can be analytically
continued (in this sense) to a Riemannian space (of Euclidean signature), then that
spacetime is characterised by its invariants.  Moreover, if a spacetime is not characterised by its
invariants, then there exists no such analytical continuation of it to a Riemannian space.

Finally, we present a summary of the Lorentzian signature case. 
All of the results presented here are  
illustrated in the 4d case. We also
discuss the 4d pseudo-Riemannian  case of neutral 
signature $(--++)$ (NS space).

\subsection{Metrics characterised by their invariants}  
Consider the continuous metric deformations defined as follows \cite{inv}.  
\begin{defn}  
\noindent For a spacetime $(\mathcal{M},g)$, a (one-parameter) \emph{metric deformation}, $\hat{g}_\tau$, $\tau\in [0,\epsilon)$, is a family of smooth metrics on $\mathcal{M}$ such that  
\begin{enumerate}  
\item{} $\hat{g}_\tau$ is continuous in $\tau$,  
\item{} $\hat{g}_0=g$; and  
\item{} $\hat{g}_\tau$ for $\tau>0$, is not diffeomorphic to $g$.  
\end{enumerate}  
\end{defn}  
\noindent For any given spacetime $(\mathcal{M},g)$ we define the set of all scalar polynomial curvature invariants  
\begin{equation}  
\mathcal{I}\equiv\{R,R_{\mu\nu}R^{\mu\nu},C_{\mu\nu\alpha\beta}C^{\mu\nu\alpha\beta}, R_{\mu\nu\alpha\beta;\gamma}R^{\mu\nu\alpha\beta;\gamma}, R_{\mu\nu\alpha\beta;\gamma\delta}R^{\mu\nu\alpha\beta;\gamma\delta},\dots\} \,. \nonumber  
\end{equation}  
\noindent Therefore, we can consider the set of invariants as a function of the metric  
and its derivatives. However, we are interested in to what extent,  
or under what circumstances, this function has an inverse.  
More precisely, define the function $\mf{I}:g_{\mu\nu}\mapsto \mathcal{I}$ 
(the $\mf{I}$-map) to be the function that calculates the set of 
invariants $\mathcal{I}$ from a metric $g_{\mu\nu}$. Given a set 
of invariants, $\mathcal{I}_0$, what is the nature of the (inverse)
set $\mf{I}^{-1}(\mathcal{I}_0)$?  
In order to address this question, we first need to introduce some terminology \cite{inv}. 
\begin{defn}  
Consider a spacetime $(\mathcal{M},g)$ with a set of invariants  
$\mathcal{I}$. Then, if there does not exist a metric deformation  
of $g$ having the same set of invariants as $g$, then we will call  
the set of invariants \emph{non-degenerate}. Furthermore, the  
spacetime metric $g$, will be called  
\emph{$\mathcal{I}$-non-degenerate}.  
\end{defn}  
  
This implies that for a metric which is  
$\mathcal{I}$-non-degenerate the invariants characterize the  
spacetime uniquely, at least locally, in the space of (Lorentzian)  
metrics. This means that these metrics are characterized by their  
curvature invariants  and therefore we can distinguish such metrics  
using their invariants. Since scalar curvature invariants are  
manifestly diffeomorphism-invariant we can thereby avoid the  
difficult issue whether a diffeomorphism exists connecting two  
spacetimes.  
 
In the above sense of $\mathcal{I}$-non-degeneracy,
a spacetime is completely characterized by its 
invariants and there is only one (non-isomorphic) spacetime with this set of 
invariants (at least locally in the space of metrics, in the sense above). However, we may also want to use the notion of characterization by 
scalar invariants in a different (weaker) sense
(see below).
In order to emphasise the difference between such cases, we will say that
$\mathcal{I}$-non-degenerate metrics are characterised by their invariants \emph{in the strong
sense}, or in short, \emph{strongly} characterised by their invariants
(this is the general sense, and if the clarifier
`strong sense' is omitted, this meaning is implied).

However, in the definition \ref{def:char} (below)
it is useful to regard a pseudo-Rie\-mannian manifold as being
characterised by its invariants if the curvature tensors 
(i.e., the Riemann tensor and all of its covariant derivatives)
can be reconstructed from a knowledge  all of the scalar 
invariants.
It is clear that $\mathcal{I}$-non-degenerate metrics fall into this class, but
another important class of metrics will also be included.  For example, (anti-)de Sitter space is
maximally symmetric and consequently has only one independent curvature component, and also
only one independent curvature invariant (namely the Ricci scalar).  Therefore, we can determine
the curvature tensors of (anti-)de Sitter space by knowing its invariants.  In our definition,
(anti-)de Sitter space is not $\mathcal{I}$-non-degenerate but it will be 
regarded as being characterised by its
invariants.

Therefore, we shall say that metrics that are
characterised by their invariants but are not $\mathcal{I}$-non-degenerate, are \emph{weakly}
characterised by their invariants.  A common factor for metrics being characterised by their
invariants is that the curvature tensors, and hence the spacetime, will inherit certain
properties of the invariants.  In particular, as we will see later, all CSI spacetimes being
characterised by their invariants will be automatically be locally homogeneous; however, not
all of them are $\mathcal{I}$-non-degenerate, like (anti-)de Sitter space.

The definition we will adopt is closely related to the familiar problem of diagonalising
matrices.  We will use curvature operator to relate our problem to linear algebra, and
essentially, we will say that a spacetime is characterised by its invariants if we can
diagonalise its curvature operators.  In order to make this definition a bit more rigourous, we
need to introduce some formal definitons.

\section{Definition} 
Consider an even-ranked tensor $T$. By raising or lowering the indices appropriately, we can construct the tensor with components $T^{\alpha_1...\alpha_k}_{\phantom{\alpha_1...\alpha_k}\beta_1...\beta_k}$. This tensor can be considered as an operator (or an endomorphism) mapping contravariant tensors onto contravariant tensors  
\[ {\sf T}: V\mapsto V,\]   
where $V$ is the vectorspace $V=(T_pM)^{\otimes k}(\equiv\bigotimes_{i=1}^kT_pM)$, or, if the tensor possess  
index symmetries, $V\subset(T_pM)^{\otimes k}$. Here, since $V$ is a vectorspace, there exists a set of basis vectors ${\bf e}_{I}$ so that $V=\mathrm{span}\{e_{I}\}$. Expressing ${\sf T}$ in this basis, we can write the operator in component form: $T^I_{~J}$. This means that we have a natural matrix representation of the operator ${\sf T}$.   
  
Note that we can similarly define a dual operator ${\sf T}^*:V^*\mapsto V^*$ where  
$V^*=(T^*_pM)^{\otimes k}$, mapping covariant tensors onto covariant tensors.  
In component form the dual operator is $T^{*I}_{\phantom{*}~J}$ which in matrix  
form can be seen as the transpose of ${\sf T}$; i.e., ${\sf T}^*={\sf T}^T$. Consequently,  
there is a natural isomorphism between the operator ${\sf T}$ and its dual. We can also consider operators mapping mixed tensors onto mixed tensors.    
  
For an odd-ranked tensor, $S$, we can also construct an operator: however, this time we need to 
consider the tensor product of $S$ with itself, $T=S\otimes S$, which is of even rank.  We 
can also construct an operator with two odd-ranked tensors $S$ and $S'$:  $T=S\otimes S'$.  In 
this way $ T$ is even-ranked and we can lower/raise components appropriately. 
 
Now, consider an operator.  If $T^I_{~J}$ are the components in a certain basis, then this as 
a natural matrix representation.  This is an advantage  because now we can use standard 
results from linear algebra to estabish some related operators.  
 
The utility of these operators when it comes to invariants is obvious since any 
curvature invariant is always an invariant constructed from some operator.  For example, the Kretchmann 
invariant, $R^{\alpha\beta\mu\nu}R_{\alpha\beta\mu\nu}$, is trivially (proportional to) the 
trace of the operator ${\sf 
T}=(T^A_{~B})=\left(R^{\alpha_1\alpha_2\alpha_3\alpha_4}R_{\beta_1\beta_2\beta_3\beta_4}\right)$. 
Similarly, an arbitrary invariant, 
$T^{\alpha_1\alpha_2\cdots\alpha_k}_{\phantom{\alpha_1\alpha_2\cdots\alpha_k}\alpha_1\alpha_2\cdots\alpha_k}$, 
is the trace of the operator ${\sf 
T}=(T^A_{~B})=\left(T^{\alpha_1\alpha_2\cdots\alpha_k}_{\phantom{\alpha_1\alpha_2\cdots\alpha_k}\beta_1\beta_2\cdots\beta_k}\right)$. 
Therefore, it is clear that we can study the polynomial curvature invariants by studing the 
invariants of curvature operators. 
 
The archetypical example of a curvature operator is the Ricci operator, ${\sf R}=(R^\mu_{~\nu})$, which maps vectors onto vectors; i.e., 
\[ {\sf R}: T_pM\mapsto T_pM.\] 
All of the Ricci invariants can be constructed from the invariants of this operator, for example, $\Tr({\sf R})$ is the Ricci scalar and $\Tr({\sf R}^2)=R_{\mu\nu}R^{\mu\nu}$. 

Another commonly used operator is the Weyl operator, ${\sf C}$, which maps bivectors onto bivectors:
\[ {\sf C}:\wedge^2T_pM\mapsto \wedge^2T_pM.\]
The Weyl invariants can also be constructed from this operator by considering traces of powers of ${\sf C}$: $\Tr({\sf C}^n)$.

\section{Eigenvalues and projectors}  
For a curvature operator, ${\sf T}$, consider an eigenvector ${\sf   
v}$ with eigenvalue $\lambda$; i.e., ${\sf T}{\sf v}=\lambda{\sf   
v}$. Note that the symmetry group 
$SO(d,n)$ of the spacetime (of signature $(d,n)$), is naturally imbedded through the tensor 
products $(T_pM)^{\otimes k}$.  By considering the eigenvalues of $\T$ as solutions of the 
characteristic equation: \[ {\mathrm{det}}(\T -\lambda{\sf 1})=0, \] which are 
$GL(k^n,\mb{C})$-invariants.  Since the orthogonal group $SO(d,n-d)$, using the tensor 
product, acts via a representation $\Gamma:  SO(d,n-d)\mapsto GL(k^n)\subset GL(k^n,\mb{C})$, 
Hence, \emph{the eigenvalue of a curvature operator is an $O(d,n-d)$-invariant curvature scalar}.   
Therefore, curvature operators naturally provide us with a set of   
curvature invariants (not necessarily polynomial invariants but derivable from them)   
corresponding to the set of distinct eigenvalues: $\{\lambda_A   
\}$. Furthermore, the set of eigenvalues are uniquely determined   
by the polynomial invariants of ${\sf T}$ via its characteristic   
equation. The characteristic equation, when solved, gives us the   
set of eigenvalues, and hence these are consequently determined by   
the invariants.

We can now define a number of associated curvature operators. For   
example,  for an eigenvector ${\sf v}_A$ so that ${\sf T}{\sf   
v}_A=\lambda_A{\sf v}_{A}$, we can construct the annihilator   
operator:   
\[ {\sf P}_A\equiv ({\sf T}-\lambda_{A}{\sf 1}).\]   
Considering the Jordan block form of ${\sf T}$, the eigenvalue ${\lambda_A}$ corresponds to a set of Jordan blocks. These blocks are of the form:   
\[ {\sf B}_A=\begin{bmatrix}   
\lambda_A & 0 & 0& \cdots & 0 \\   
1 & \lambda_A & 0& \ddots  & \vdots \\   
0      & 1 &\lambda_A& \ddots & 0 \\   
\vdots & \ddots     &\ddots& \ddots & 0 \\   
0    &     \hdots &0   &  1      & \lambda_A   
\end{bmatrix}.\]   
There might be several such blocks corresponding to an eigenvalue   
$\lambda_A$; however, they are all such that $({\sf   
B}_A-\lambda_A{\sf 1})$ is nilpotent and hence there exists an   
$n_{A}\in \mathbb{N}$ such that  ${\sf P}_A^{n_A}$ annihilates the   
whole vector space associated with the eigenvalue $\lambda_A$.   
   
This implies that we can define a set of operators $\widetilde{\bot}_A$ with eigenvalues $0$ or $1$ by considering the products   
\[ \prod_{B\neq A}{\sf P}^{n_B}_B=\Lambda_A\widetilde{\bot}_A,\]   
where $\Lambda_A=\prod_{B\neq A}(\lambda_A-\lambda_B)^{n_B}\neq 0$   
(as long as $\lambda_B\neq \lambda_A$ for all $B$). Furthermore, we can now define  
\[ \bot_A\equiv {\sf 1}-\left({\sf 1}-\widetilde{\bot}_A\right)^{n_A}  \]  
 where $\bot_A$   
is a \emph{curvature projector}. The set of all such curvature   
projectors obeys:   
\beq {\sf 1}=\bot_1+\bot_2+\cdots+\bot_A+\cdots,   
\quad \bot_A\bot_B=\delta_{AB}\bot_A.   
\eeq We can use these   
curvature projectors to decompose the operator ${\sf T}$:   
\beq   
{\sf T}={\sf N}+\sum_A\lambda_A\bot_A. \label{decomp}   
\eeq   
The   
operator ${\sf N}$ thus contains  all the information not   
encapsulated in the eigenvalues $\lambda_A$. From the Jordan form   
we can see that ${\sf N}$ is nilpotent; i.e., there exists an   
$n\in\mathbb{N}$ such that ${\sf N}^n={\sf 0}$. In particular, if   
${\sf N}\neq 0$, then ${\sf N}$ is a negative/positive   
boost weight operator which can be used to lower/raise the   
boost weight of a tensor \cite{class,bivect}.
  
We also note that   
\[ {\sf N}=\sum_A {\sf N}_A, \quad {\sf N}_A\equiv \bot_A {\sf N}\bot_A.\]   
 Consequently, we have the orthogonal decomposition:   
\beq   
{\sf T}=\sum_A\left({\sf N}_A+\lambda_A\bot_A\right). \label{decomp2}   
\eeq   
Note that we can achieve this decomposition by just using operators and their invariants.   
  
In linear algebra we are accustomed to the concept of diagonalisable matrices.  
Similarly, we will call an operator  \emph{diagonalisable} if   
\beq   
{\sf T}=\sum_A\lambda_A\bot_A.   
\label{Tdiag}\eeq   
  
In differential geometry we are particularly interested in quantities like the Riemann  
curvature tensor, $R$. Now, the Riemann tensor naturally defines two curvature operators:  
namely, the Ricci operator ${\sf R}$, and the Weyl operator ${\sf C}$.  
Let us, at a point $p$, define the tensor algebra (or tensor concomitants) of 0th order curvature operators,  
$\mathcal{R}_0$,\footnote{Here, $\mathcal{R}$ refers to the Riemann tensor and index $0$ refers to the number of covariant derivatives. } as follows:  
\begin{defn}  
The set ${\mathcal R}_0$ is a $\mathbb{C}$-linear space of operators (endomorphisms) which satisfies the following properties: 
\begin{enumerate} 
\item{} Identity: the metric tensor, as an operator, is in  $\mathcal{R}_0$. 
\item{} the Riemann tensor, as an operator, is in $\mathcal{R}_0$. 
\item{} if ${\sf T}\in\mathcal{R}_0$, then any tensor contraction and any eigenvalue of ${\sf T}$ is also in $\mathcal{R}_0$. 
\item{} ``Square roots'': if $\tilde{{\sf T}}\in\mathcal{R}_0$ and $\tilde{ T}={ T}\otimes{ T}$ (as tensors) where $T$ is of even rank,  then ${\sf T}\in \mathcal{R}_0$ 
\item{} any operator that can be considered as a tensor product or expressible as functions of elements in $\mathcal{R}_0$ is also in $\mathcal{R}_0$.  
\end{enumerate}  
An element of $\mathcal{R}_0$, will be referred to as a \emph{curvature operator of order 0}.   
\end{defn}  
Similarly, allowing for contractions of the covariant derivative of the Riemann tensor, $\nabla R$,  
we can define the tensor algebra of 1st order curvature operators, $\mathcal{R}_1$, etc. If  
we allow for arbitrary number of derivatives, we will simply call it $\mathcal{R}$. Hence\footnote{As a $\mathbb{C}$-linear vector space, $\mathcal{R}_k$ will be finite dimensional, while $\mathcal{R}$ is infinite dimensional. However, in practice, it is sufficient to consider $\mathcal{R}_k$ for sufficiently large $k$ due to the fact that a finite number of derivatives of the Riemann tensor are sufficient to determine the spacetime up to isometry.}:  
\[ \mathcal{R}_0\subset\mathcal{R}_1\subset \cdots \subset \mathcal{R}_k\subset \mathcal{R}.\]  
For example, the following tensor:  
\[ R^{\mu\nu\alpha\beta}R_{\alpha\beta\rho\sigma}+45R^{\alpha\beta\gamma\delta;\mu\nu}R_{\alpha\beta\epsilon\eta}R^{\epsilon\eta}_{~\phantom{\epsilon\eta}\gamma\delta;\rho\sigma},\]  
is a curvature tensor of order 2 and is consequently an element of $\mathcal{R}_2$.   
  
Of particular interest is when all curvature operators can be determined using their invariants.  
This means that there is preferred set of operators (namely the projectors) which acts as  
a basis for the curvature operators. This is related to diagonalisability of operators, and hence:  
\begin{defn}  
$\mathcal{R}_k$ is called diagonalisable if there exists a set of projectors $\bot_A\in \mathcal{R}_k$  forming a tensor basis for $\mathcal{R}_k$; i.e., for every ${\sf T}\in \mathcal{R}_k$,   
\[{\sf T}=T^{AB...D}\bot_A\otimes\bot_B\otimes\cdots\otimes\bot_D.\]   
\end{defn}  
Note that the components may be complex. If an operator is diagonalisable, we can, by solving the characteristic equation,  
determine the expansion eq. (\ref{Tdiag}). This enables us to reconstruct the operator itself.  
In this sense, the operator would be characterised by its invariants. Hence, we have the following definition:  
\begin{defn} \label{def:char} 
A space $(\mathcal{M},g_{\mu\nu})$ is said to be \emph{characterised by its invariants} iff, for every $p\in \mathcal{M}$, the set of curvature operators, $\mathcal{R}$, is diagonalisable (over $\mb{C}$).   
\end{defn}  
As pointed out earlier we can further subdivide this category  into spaces that are characterised by their invariants in a strong or weak sense. If, in addition, the space $(\mathcal{M},g_{\mu\nu})$ is $\mathcal{I}$-non-degenerate, then the space is characterised by its invariants in the \emph{strong} sense; otherwise its said to be  characterised by its invariants in the \emph{weak} sense.  
 
A word of caution:  note that even though we know all the operators of a spacetime that is characterised by 
its invariants, we do not necessarily know the frame.  In particular, by raising an index of the metric 
tensor, $g_{\mu\nu}$, we get, irrespective of the signature, $g^{\mu}_{~\nu}=\delta^{\mu}_{~\nu}$.  This 
means that for a Lorentzian spacetime we would lose the information of which direction is time.  However, 
there might still be ways of determining which is time by inspection of the invariants (but this is no 
guarantee). If, say, an invariant can be written as $I=r^{\alpha}r_{\alpha}$, and $I<0$, then clearly, $r^\alpha$ is timelike and, consequently, can be used as time.  On the other hand, an example of when this cannot be done is the following case, in $d$ dimensions, when the only non-zero invariants 
are the 0th order Ricci invariants:   
\[ \Tr({\sf R})=R^{\mu}_{~\mu}=(d-1)\lambda, \quad \Tr({\sf 
R}^n)=R^{\mu_1}_{~\mu_2}R^{\mu_2}_{~\mu_3}\cdots R^{\mu_{n}}_{~\mu_1}=(d-1)\lambda^n.\] Then, if this 
spacetime is characterised by its invariants, we get:  \[ {\sf R}=\diag(0,\lambda,\lambda,\cdots,\lambda)\] 
However, there is no information in the invariants if the direction associated with the $0$ eigenvalue is a 
space-like or time-like direction.  Consequently, in the definition above we have to keep in mind that 
these spacetimes are characterised by their invariants, up to a possible ambiguity in which direction is 
associated with time, and which is space\footnote{Note that in the Riemannian signature case, there is no 
such ambiguity since all directions are necessarily space-like.}. 
 
In the following we will reserve the word ``spacetime'' to the Lorentzian-signature case, 
while ``Riemannian'' (space) will correspond to the Euclidean signature case.  For the Riemannian case, the operators provide us with an extremely simple proof of the following theorem:  
\begin{thm}\label{thm:riemannian}  
A Riemannian space is always characterised by its invariants.  
\end{thm}  
\begin{proof}  
In an orthonormal frame, the Riemannian metric is $g_{\mu\nu}=\delta_{\mu\nu}$. Therefore, the components  
of an operator $T^I_{~J}$ will be the same as those of $T_{IJ}$, possibly up to an overall constant. Therefore, let us decompose into the symmetric, 
$S_{IJ}$, and antisymmetric part $A_{IJ}$, of $T_{IJ}$:   
\[ T_{IJ}=S_{IJ}+A_{IJ}, \quad S_{IJ}\equiv T_{(IJ)}, \quad A_{IJ}\equiv T_{[IJ]}.\]  
Since the signature is Euclidean, the operators $S^{I}_{~J}$ and $A^{I}_{~J}$ will also be symmetric and  
antisymmetric (as matrices), respectively. A standard result is that we can consequently always diagonalise ${\sf S}$:    
\beq  
{\sf S}=\diag(\lambda_1,\lambda_2,\cdots,\lambda_N);  
\eeq  
i.e., the operator ${\sf S}$ is diagonalisable. For ${\sf A}$:   
\beq  
{\sf A}=\mathrm{blockdiag}\left(\begin{bmatrix} 0 & a_1 \\ -a_1 & 0 \end{bmatrix},\begin{bmatrix} 0 & a_2 \\ -a_2 & 0 \end{bmatrix},\cdots,\begin{bmatrix} 0 & a_k \\ -a_k & 0 \end{bmatrix}, 0,\cdots, 0  \right)  
\eeq  
Consequently, ${\sf A}$ is diagonalisable also (with purely imaginary eigenvalues).   
  
So for ${\sf T}\in\mathcal{R}$, then also ${\sf S}\in\mathcal{R}$ and ${\sf A}\in\mathcal{R}$, and therefore any curvature operator is characterised by its invariants. This proves the theorem.  
\end{proof}  
 
We note that this theorem also follows from a group-theoretical perspective 
\cite{PTV,Procesi}; however, the 
operators provide an alternative (and fairly straight-forward) proof. 
 
Of course, it is known that a Lorentzian spacetime is not necessarily characterised by its invariants 
in any dimension.\footnote{In the proof of theorem \ref{thm:riemannian} it can be seen that what goes wrong is that $S_{IJ}$ and $A_{IJ}$ are not necessarily symmetric or antisymmetric as operators $S^I_{~J}$ and $A^I_{~J}$. Therefore, the Jacobi canonical forms are really needed in this case.}   
However, the spacetimes being characterised by their curvature invariants play a special role for certain  
classes of metrics. For example, we have that:  
\begin{prop}[VSI spacetimes]  
Among the spacetimes of vanishing scalar curvature invariants (VSI), the only spacetimes being characterised by their curvature invariants (weakly or strongly) are locally isometric to flat space.  
\end{prop}  
Note that flat space is a characterised by its invariants only in the weak sense.   
 
\begin{prop}[CSI spacetimes]  
If a spacetime has constant scalar curvature invariants (CSI) and is characterised by its invariants (weakly or strongly), then it is a locally homogeneous space.  
\end{prop}  
These locally homogeneous spacetimes can be characterised by their invariants in either a weak or strong sense.  
 
Regarding the question of which spacetimes are characterised by its invariants, in the sense defined above, we have that the following conjecture is true:  
\begin{con}  
If a spacetime is characterised by its invariants (weakly or strongly), then it is either  
 $\mathcal{I}$-non-degenerate, or of type D to all orders; i.e., type D$^k$.  
\end{con}  
There are several results that support this conjecture.  First, it seems reasonable that all 
$\mathcal{I}$-non-degenerate metrics are (strongly) characterised by their invariants.  In 4 
dimensions, this seems clear from the results of \cite{inv}.  Furthermore, that type D$^k$ spacetimes 
are also characterised by its invariants (in a weak sense necessarily) follows from the results of \cite{shortinv}.  However, it 
remains to prove that these are all. 
 
We should also point out a sublety in the definition of what we mean by characterised by the invariants.  
Now a class of metrics that, in general, does not seem to be characterised by their curvature invariants  
is the subclass of Kundt metrics \cite{inv,Kundt}:   
\[ \d s^2=2\d u[\d v+v^2H(u,x^k)\d u+vW_i(u,x^k)\d x^i]+g_{ij}(u,x^k)\d x^i \d x^j. \]  
For sufficiently general functions $H$, $W_i$, and $g_{ij}$, these metrics  
will not be characterised by their invariants as defined above because they will be of type II  
to all orders (in three dimensions, see the example below). However, the invariants still enable  
us to reconstruct the metric. In some sense, the metric above is the simplest metric with this set of  
invariants (however, it is not characterised by its invariants as defined above!). 
\paragraph{Example: 3D degenerate Kundt metrics.} Let use study the details of the above example  
in three dimensions, but generalising to the full class of degenerate Kundt metrics, which can be written:  
\[ \d s^2=2\d u[\d v+H(v,u,x)\d u+W(v,u,x)\d x]+\d x^2, \]  
where $H=v^2H^{(2)}(u,x)+vH^{(1)}(u,x)+H^{(0)}(u,x)$ and  $W=vW^{(1)}(u,x)+W^{(0)}(u,x)$.  
In 3D the Weyl tensor is zero and so we only need to consider the Ricci tensor which, using  
the coordinate basis $(u,v,x)$, can be written in operator form:   
\beq  
{\sf R}=\begin{bmatrix}   
\lambda_1 & 0 & 0 \\  
R_{uu} & \lambda_1 & R_{ux} \\  
R_{ux} & 0 & \lambda_3  
\end{bmatrix}.  
\eeq  
The eigenvalues are $\lambda_1, \lambda_1, \lambda_3$; consequently,  
if $\lambda_1\neq \lambda_3$ the Segre type is $\{21\}$ or $\{(1,1)1\}$. Unfortunately,  
the coordinate basis is not the canonical Segre basis so the form is not manifest. However,  
the projection operators are frame independent operators so let us find the various projectors  
in this case. Since $\lambda_1$ has multiplicity 2, (while $\lambda_3$ has multiplicity 1) the  
operator $({\sf R}-\lambda_1{\sf 1})^2$ must be proportional to the projection operator $\bot_3$. By ordinary matrix multiplication we get:  
\beq  
({\sf R}-\lambda_1{\sf 1})^2=\begin{bmatrix}   
0 & 0 & 0 \\  
R_{ux}^2 & 0 & (\lambda_3-\lambda_1)R_{ux} \\  
(\lambda_3-\lambda_1)R_{ux} & 0 &(\lambda_3-\lambda_1)^2  
\end{bmatrix}.  
\eeq  
The proportionality constant is $(\lambda_3-\lambda_1)^2$. Consequently:  
\[ \bot_3=\begin{bmatrix}   
0 & 0 & 0 \\  
\tfrac{R_{ux}^2}{(\lambda_3-\lambda_1)^2} & 0 & \tfrac{R_{ux}}{(\lambda_3-\lambda_1)} \\  
\tfrac{R_{ux}}{(\lambda_3-\lambda_1)} & 0 & 1  
\end{bmatrix}.  
\]  
We note that the projection operator is not diagonal in the coordinate basis: however, we can easily verify 
that $\bot_3^2=\bot_3$. 
 
There are only two projection operators in this case, so $\bot_1$ can be calculated  
using $\bot_1={\sf 1}-\bot_3$. Alternatively, we can define  
$\widetilde{\bot}_1\equiv(\lambda_1-\lambda_3)^{-1}({\sf R}-\lambda_3{\sf 1})$ so that the projection  
operator is given by $\bot_1={\sf 1}-({\sf 1}-\widetilde{\bot}_1)^2$. Using either  
way to calculate $\bot_1$, we obtain:  
 \[ \bot_1=\begin{bmatrix}   
1 & 0 & 0 \\  
-\tfrac{R_{ux}^2}{(\lambda_3-\lambda_1)^2} & 1 & -\tfrac{R_{ux}}{(\lambda_3-\lambda_1)} \\  
-\tfrac{R_{ux}}{(\lambda_3-\lambda_1)} & 0 & 0  
\end{bmatrix}.  
\]  
To check which Segre type the Ricci tensor has we calculate the nilpotent operator given by the   
expansion eq. (\ref{decomp}):  
\beq  
{\sf N}={\sf R}-\lambda_1\bot_1-\lambda_3\bot_3=\begin{bmatrix}   
0 & 0 & 0 \\  
R_{uu}-\tfrac{R_{ux}^2}{(\lambda_3-\lambda_1)} & 0 & 0 \\  
0 & 0 & 0  
\end{bmatrix}.  
\eeq  
Therefore, the metric is Segre type $\{21\}$ in general, while if $R_{uu}=\tfrac{R_{ux}^2}{(\lambda_3-\lambda_1)}$ then it is Segre type $\{(1,1)1\}$.   
  
For the above metric we have   
\beq  
\lambda_1&=& 2H^{(2)}+\frac{1}{2}\frac{\partial W^{(1)}}{\partial x}-\frac 12 (W^{(1)})^2, \\  
\lambda_3&=& \frac{\partial W^{(1)}}{\partial x}-\frac 12 (W^{(1)})^2.   
\eeq  
The invariants will only depend on $H^{(2)}$ and $W^{(1)}$,  
and the invariants will determine these functions up to diffeomorphisms. However,  
in general, this metric will be of Segre type $\{2 1\}$; therefore, even if $W^{(0)}=H^{(1)}=H^{(0)}$,  
this metric will not be determined by the curvature invariants, as defined above.  
  
\section{The operator calculus}   
Since all of the results so far are entirely point-wise, it is also an advantage to consider  
the operator calculus; i.e. derivatives of operators. Most useful for our purposes is the Lie derivative.  
The Lie derivative preserves the order and type of tensors and is thus particularly useful.   
  
Consider a vector field ${\mbold\xi}$ defined on a neighbourhood $U$. We must assume that \emph{the operator decomposition, (\ref{decomp2}), does not change over $U$.} This assumption is essential in what follows.   
  
Consider the one-parameter group of diffeomorphisms, $\phi_t$, generated by the vector field ${\mbold\xi}$. Assume that two points $p$ and $\hat{p}$ (both in $U$) are connected via $\phi_t$; i.e., $\hat{p}=\phi_t(p)$. Since the eigenvalues are scalar functions over $U$, then $\phi_t^*(\lambda)=\hat{\lambda}$. Furthermore, by assumption, the eigenvalue structure of ${\sf T}$ does not change over $U$; consequently, for a eigenvector ${\sf v}$ with eigenvalue $\lambda$ we have:  
\[ \phi^*_t({\sf T}{\sf v}-\lambda{\sf v})=\hat{\sf T}\hat{\sf v}-\hat{\lambda}\hat{\sf v}={\sf T}{\sf v}-\lambda{\sf v}=0.\]   
This implies that eigenvectors are mapped onto eigenvectors. Therefore, if the eigenvalue  
$\lambda$ is mapped onto $\hat{\lambda}$ (these are scalar functions determined by the characteristic equation), then the eigenvector ${\sf v}$ is mapped onto the corresponding eigenvector $\hat{\sf v}$. Then if ${\bf e}_I$ spans the eigenvectors of eigenvalue $\lambda$, and $\hat{\bf e}_I$ spans the eigenvectors with eigenvalue $\hat{\lambda}$, then there must exist, since $\phi_t$ also preserves the norm, an invertible matrix $M^I_{~J}$ so that:  
\[ \hat{\bf e}_J=M^I_{~J}{\bf e}_I.\]   
If ${\mbold\omega}^I$ and $\hat{\mbold\omega}^I$ are the corresponding one-forms, then $\hat{\mbold\omega}_J=(M^{-1})^I_{~J}{\mbold\omega}_I$.  
For a projector, the eigenvalues are $\lambda = 0,1$; therefore, we can write  
\[ \bot = \delta^{J}_{~I}{\bf e}_J\otimes {\mbold\omega}^I\]   
where the eigenvectors ${\bf e}_J$ have eigenvalue 1.   
Therefore:  
\[ \phi^*_t(\bot)\equiv\hat{\bot}_t=\bot.\]   
Consequently,   
\[ \pounds_{\mbold\xi}\bot = 0.\]   
Therefore, we have the remarkable property that the projectors are Lie transported with respect to  
any vector field ${\mbold\xi}$.   
  
For the Lie derivative we thus have the following result:  
\begin{thm}  
Assume that the operator, ${\sf T}$, has the decomposition (\ref{decomp2}). Then, the Lie derivative with respect to ${\mbold\xi}$ is  
\[ \pounds_{\mbold\xi}{\sf T}=\sum_A\left[\pounds_{\mbold\xi}{\sf N}_A+{\mbold\xi}\left(\lambda_A\right)\bot_A\right], \quad \text{where}\quad \pounds_{\mbold\xi}{\sf N}_A=\bot_A(\pounds_{\mbold\xi}{\sf N})\bot_A. \]  
\end{thm}  
\begin{proof}  
We have  
\[ \pounds_{\mbold\xi}\left(\sum_A \lambda_A\bot_A\right)=\sum_A\left[\left(\pounds_{\mbold\xi}\lambda_A\right)\bot_A+\lambda_A\pounds_{\mbold\xi}(\bot_A)\right]=\sum_A{\mbold\xi}(\lambda_A)\bot_A, \]   
and   
\beq \pounds_{\mbold\xi}\left(\bot_A{\sf N}\bot_B\right)& =&\left(\pounds_{\mbold\xi}\bot_A\right){\sf N}\bot_B+\bot_A\left(\pounds_{\mbold\xi}{\sf N}\right)\bot_B+\bot_A{\sf N}\left(\pounds_{\mbold\xi}{\bot_B}\right)\nonumber \\  &=&\bot_A\left(\pounds_{\mbold\xi}{\sf N}\right)\bot_B\nonumber.   
\eeq   
Since  $\bot_A{\sf N}\bot_B=\delta_{AB}{\sf N}_A$, the theorem now follows.  
\end{proof}  
Now clearly, this has some interesting consequences. If, for example, the operator ${\sf T}$ is diagonalisable, i.e., ${\sf N}=0$, then ${\sf T}$ is Lie transported if and only if the eigenvalues are Lie transported. This can formulated as follows:  
\begin{cor}   
If an operator, {\sf T}, is diagonalisable, then, for a vector field, ${\mbold\xi}$,    
\[ \pounds_{\mbold\xi}{\sf T}=0\quad \Leftrightarrow\quad {\mbold\xi}(\lambda_A)=0\]   
\end{cor}  
\begin{cor}  
If a spacetime $(\mathcal{M},g_{\mu\nu})$ is characterised by its invariants (weakly or strongly), then if there exists  
a vector field, ${\mbold\xi}$, such that 
\[ {\mbold\xi}(I_i)=0. \quad \forall I_i\in \mathcal{I},\]    
then there exists a set, $\mathcal{K}$, of  Killing vector fields such that at any point the vector field ${\mbold\xi}$ coinsides with a Killing vector field $\tilde{\mbold\xi}\in \mathcal{K}$.  
\end{cor}  
  
The last corollary enables us to reduce the question of Killing vectors down to the existence of vectors annihilating all the curvature invariants $I_i$. This may be easier in some cases if the form of the metric is totally unmanageable.   
  
\subsection*{Examples}  
  
\paragraph{CSI spacetimes:}  
  
Consider now a spacetime which has all constant scalar curvature invariants \cite{CSI,3CSI,4CSI}. Assume also that the spacetime is characterised by its invariants, weakly or strongly. Let us now see how we can give an alternative proof that this must be locally homogeneous.   
  
We choose, at any point $p$, a local coordinate system $\{x^k\}$. Using the coordinate frame, we can define ${\mbold\xi}_k=\partial_k$ we get:   
\[ {\mbold\xi}_k(I_i)=0, \quad \forall I_i\in\mathcal{I}.\]  
Therefore, for each ${\mbold\xi}_k$ (which are linearly independent at $p$), there would be a Killing vector field $\tilde{\mbold\xi}_k$ coinsiding with ${\mbold\xi}_k$ at $p$. Therefore, the Killing vectors  are also linearly independent and span the tangent space, consequently, the spacetime is \emph{locally homogeneous}.   
  
\paragraph{Kinnersley class I vacuum metrics:}  
Let us consider the Kinnersley class I Petrov type D vacuum metric \cite{Kinnersley} which is a $\mathcal{I}$-non-degenerate metric \cite{typeD}. The Cartan invariants can all be reconstructed from the 4 (complex) scalar polynomial invariants:  
\begin{eqnarray}  
I &=& \frac 12\Psi_{abcd}\Psi^{abcd}=3(2Cil+m)^2 z^{-6}, \nonumber \\        
C^{\alpha\beta\gamma\delta}C_{\alpha\beta\gamma\delta}&=& 24(2Cil+m)^2 z^{-6} +24(-2Cil+m)^2 z^{*-6}\nonumber \\  
{\Psi}^{(abcd;e)f'} {\Psi}_{(abcd;e)f'}  &=& 180(2Cil+m)^2 Sz^{-8}\nonumber \\                          
C^{\alpha\beta\gamma\delta;\mu}C_{\alpha\beta\gamma\delta;\mu}&=& 720(2Cil+m)^2 Sz^{-8} +720(-2Cil+m)^2 Sz^{*-8}\nonumber,   
\end{eqnarray}  
where $\Psi_{abcd}$ is the Weyl spinor, $a,b,..$ are spinor indices, and $\alpha, \beta,...$ are  
frame indices. Here, $C$,  $m$ and $l$ are all constants while all of the functions ($S$ etc.) 
depend only on the coordinate $r$.   
  
Since the invariants only depends on one variable, namely $r$, this spacetime possesses (at least) 3 transitive Killing vectors at any point.  
  
%\paragraph{Kerr metric:}   
%For the Kerr metric  
%\[ \Psi_2=-\frac{M}{(r-ia\cos\theta)^3}\]   

\section{Analytic metric continuation}  
  
The operators also have another property that seems to be very useful.  The metric $g_{\mu\nu}$ does not 
appear explicity in the analysis above and only appears after raising an index:  $g^{\mu}_{~\nu}=\delta^{\mu}_{~\nu}$. 
Consequently, this analysis is independent of the signature of the metric.  In the eigenvalue equation the identity 
operator ${\sf 1}$ would therefore be independent of the signature also.  We can thus consider what happens 
under more general coordinate transformations than the real diffeomorphisms preserving the metric 
signature.  Let us therefore consider complex analytic continuations of the real metric of this form. 
 
Consider a point $p$ and a neighbourhood, $U$ at $p$, and we will assume this nighbourhood is an analytic 
neighbourhood and that $x^{\mu}$ are coordinates on $U$ so that $x^{\mu}\in \mathbb{R}^n$.  We will adapt 
the coordinates so that the point $p$ is at the origin of this coordinate system.  Consider now the 
complexification of $x^{\mu}\mapsto x^{\mu}+iy^{\mu}=z^{\mu}\in \mathbb{C}^n$.  This complexification 
enables us to consider the complex analytic neighbourhood $U^{\mathbb C}$ of $p$. 
 
Furthermore, let $g_{\mu\nu}^{\mathbb C}$ be a complex bilinear form induced by the analytic extension of the metric:   
\[ g_{\mu\nu}(x^{\rho})\d x^\mu\d x^{\nu}\mapsto  g^{\mathbb C}_{\mu\nu}(z^{\rho})\d z^\mu\d z^{\nu}.\]   
Consider now a real analytic submanifold containing $p$: $\bar{U}\subset U^{\mathbb C}$ with coordinates $\bar{x}^\mu\in \mathbb{R}^n$. The imbedding $\iota:\bar{U}\mapsto U^{\mathbb C}$ enables us to pull back the complexified metric $g^{\mathbb C}$ onto $\bar{U}$:  
\beq  
\bar{g}\equiv \iota^*{g}^{\mathbb{C}}.   
\eeq   
In terms of the coordinates $\bar{x}^\mu$:  $\bar{g}=\bar{g}_{\mu\nu}(\bar{x}^\rho)\d \bar{x}^\mu\d 
\bar{x}^{\nu}$.  This bilinear form may or may not be real.  However, \emph{if the bilinear form 
$\bar{g}_{\mu\nu}(\bar{x}^\rho)\d \bar{x}^\mu\d \bar{x}^{\nu}$ is real (and non-degenerate) then we will 
call it an analytic metric extension  of $g_{\mu\nu}(x^{\rho})\d x^\mu\d x^{\nu}$ with respect to $p$.} 
 
In the following, let us call the analytic metric extension for $\bar{\phi}$; i.e., $\bar{\phi}:U\mapsto \bar{U}$. We note that this transformation is complex, and we can assume, since $U$ is real analytic, that $\bar{\phi}$ is analytic.  
  
The analytic metric continuation leaves the point $p$ stationary. Therefore, it induces a linear  
transformation, $M$, between the tangent spaces $T_pU$ and $T_p\bar{U}$. The transformation $M$ is  
complex and therefore may change the metric signature; consequently, even if the metric  
$\bar{g}_{\mu\nu}$ is real, it does not necessarily need to have the same signature as $g_{\mu\nu}$.  
  
Consider now the curvature tensors, $R$ and $\nabla^{(k)}R$ for $g_{\mu\nu}$, and $\bar{R}$ and 
$\bar{\nabla}^{(k)}\bar{R}$ for $\bar{g}_{\mu\nu}$.  Since both metrics are real, their curvature tensors 
also have to be real.  The analytic metric continuation induces a linear transformation of 
the tangent spaces; consequently, this would relate the Riemann tensors $R$ and $\bar{R}$ through a complex 
linear transformation.  However, how are these related? 
 
By using $\bar{\phi}$ we can relate the metrics $g=\bar{\phi}^*\bar{g}$.  Since the map is analytic (albeit 
complex), the curvature tensors are also related via $\bar\phi$.  For the operators this has a very useful 
consequence.  First, we note that \emph{scalar polynomial invariants are invariant under 
$GL(\mathbb{C},n)$}; therefore, the eigenvalues of ${\sf R}$ and $\bar{\sf R}$ are identical at $p$. 
Consequently, the eigenspace decomposition is identical also.  (This seems almost too remarkable to be 
true).  Over the neigbourhoods $\bar{U}$ and $U$ we can, in general, relate the eigenvalues: 
$\lambda_{A}=\bar\phi^*\bar\lambda_A$. 
 
Since, a Riemannian space is always characterised by its invariants, we immediately have the following result:   
\begin{thm} \label{analytic} 
Assume that a space, $(\mathcal{M},g_{\mu\nu})$, of any signature can be analytically continued, in the sense above, to a Riemannian space (of Euclidean signature). Then the spacetime is characterised by its invariants in either a weak or strong sense.   
\end{thm}  
  
Note that the result is actually stronger than this because not only is it diagonalisable, if we have a 
symmetric tensor giving rise to an operator (for example ${\sf R}$) then, since the Riemannian space must 
have a symmetric operator $\bar{\sf R}$, the eigenvalues are real.  Consequently, the eigenvalues at $p$ 
must also be real. 
 
Interestingly, the reverse of the above theorem is just as useful:   
\begin{cor}  
If a spacetime is not characterised by its invariants (weakly or strongly), then there exists no analytical continuation of it to a Riemannian space.   
\end{cor}  
  
In the above we have restricted to analytic continuations $\bar{\phi}$ leaving a point $p$ fixed. In general, one can consider that does not necessarily leave a point fixed. However, for such complex mappings, we need to be a bit careful with the radius of convergence. On the other hand, assuming convergence, the invariants, $I_i$, can still be related via $\bar{\phi}^*\bar{I}_i=I_i$.  Recall that we are restricting ourselves to neighbourhoods  where operators do not change algebraic form.

For spacetime metrics of algebraically special types there are also some useful 'no-go' theorems for Lorentzian manifolds:\footnote{Note that type D is by definition excluded from this Corollary.}
\begin{cor}
If a spacetime is of (proper) Weyl, Ricci or Riemann type N, III or II, then there exists no  analytical continuation of it to a Riemannian space.   
\end{cor}
\begin{proof}

This follows from the fact that \emph{any} symmetric operator of a Riemannian space is diagonalisable. Since the Weyl, Ricci and Riemann operator of type N, III and II are not diagonalisable, the corollary follows. 
\end{proof}
There are also a bundle of similar results that follows from a similar analysis.

\subsection*{Signature and convention}  
Since analytic metric continuation can change the signature, we may be in a 
situation in which we are comparing invariants of spacetimes of different signatures.   
Let us note a word of 
caution in this regard. 
 
Here, we have (implicitly) assumed that the Riemann tensor $R^{\alpha}_{~\beta\gamma\delta}$ is given in a 
coordinate basis with the standard formula involving the metric and the Christoffel symbols.  This formula 
is invariant under an overall change of signature; for example, $(++\cdots +)$ is the same as $(--\cdots 
-)$.  Consequently, the Ricci tensor, $R_{\mu\nu}=\bar{R}_{\mu\nu}$ also.  However, the sign matters for 
the Ricci scalar defined as $R=g^{\mu\nu}R_{\mu\nu}$ where we get $R=-\bar{R}$.  Therefore, even if the 
invariants may change with a sign, the overall signature is just a convention.  This is also evident from 
entities constructed from a non-zero curvature ``vector'' $r^{\alpha}$.  The norm of this vector is an 
invariant $I=r^{\alpha}r_{\alpha}$; however, in signature $(++\cdots +)$ this is necessarily positive, 
$I>0$; while signature $(--\cdots -)$ this is necessarily negative $I<0$. 
 
It is therefore crucial when comparing invariants through an analytic continuation that we specify the signature. Albeit it is conventional to say that metrics with an overall change of signature are indentical. However, the invariants of such metrics may change by a sign. 

Let us be more precise. Consider an overall change of signature: $g_{\mu\nu}\mapsto \epsilon g_{\mu\nu}$, where $\epsilon =\pm 1$. Using a coordinate basis, we see that the Christoffel symbols $\Gamma^{\alpha}_{~\mu\nu}$ do not depend on $\epsilon$. Consequently, the Riemann tensor, $R^{\alpha}_{~\beta\mu\nu}$ does not depend on $\epsilon$ either. Moveover, we can also see that the definition of the covariant derivatives, $R^{\alpha}_{~\beta\mu\nu;\lambda_1...\lambda_k}$, do not depend on $\epsilon$. However, \emph{raising} or \emph{lowering} indices introduce $\epsilon$. Therefore, depending on whether you raise/lower an odd or even number of indicies to create an invariant, the invariant will depend (if odd) or not depend (if even) on $\epsilon$. We can therefore split the invariants into invariants that depend on $\epsilon$, denoted $ I^\epsilon_i$, and those that do not, denoted $I_j$. Since \emph{the choice of overall sign is merely a convention} we make the following identification: 
\beq 
(g_{\mu\nu},~I^\epsilon_i, ~I_j)\sim  (-g_{\mu\nu},~-I^\epsilon_i, ~I_j).
\eeq
Note that in the neutral signature case changing the overall sign does not change the signature; consequently, the above map gives a non-trivial equivalence relation between invariants of neutral signature metrics. 
\subsection*{Examples}  
\paragraph{Kasner universe:}  
In 1921 E. Kasner \cite{Kasner} wrote down a Ricci flat metric of Riemannian signature (in dimension 4 but we 
present it in $n$ dimension):   
\beq  
\d s^2=\d t^2+\sum_{i=1}^nt^{2p_i}(\d x^i)^2, \qquad \sum_{i=1}^{n}p_i=\sum_{i=1}^np_i^2=1.  
\label{Kasner}\eeq  
Clearly, this has $n$ Killing vectors, $\frac{\partial}{\partial x^i}$, and therefore the  
invariants can only depend on $t$: $I=I(t)$. Furthermore, since this is a Riemannian space,  
this space is characterised by its invariants.   
  
There is a number of analytic metric continuations possible using $\bar{x}^j=ix^j$, where $i$ is the imaginary unit and $x^j$ is any of the coordinates except $t$. The so-called Kasner Universe is the Lorentzian version obtained by $(\bar{x}^1,\cdots,\bar{x}^n)=(ix^1,\cdots,ix^n)$ of signature $(+--\cdots -)$. Consequently, the invariants of this spacetime are identical to those of Riemannian version of signature $(+++\cdots+)$. In particular, this means that curvature invariants constructed from the norm of the gradient of an invariant, must be non-negative; i.e., $(\nabla_{\alpha}I)(\nabla^\alpha I)\geq 0$.   
  
Usually, the Lorentzian version is written with signature $(-++\cdots+)$, which can be obtained from  
the $(---\cdots-)$ version; consequently,    $(\nabla_{\alpha}I)(\nabla^\alpha I)\leq 0$.  
\paragraph{Schwarzschild spacetime:} The 4D Schwarzschild spacetime is given by  
\beq  
\d s^2=-\left(1-\frac{2m}r\right)\d t^2+\frac{\d r^2}{\left(1-\frac{2m}r\right)}+r^2(\d \theta^2+\sin^2\theta \d \phi^2).  
\eeq   
There exists an analytic metric continuation (not necessarily unique) to a Euclidean signature $(++++)$ given by  
\[ (\bar{t},r,{\theta},\phi)=(it,r,\theta,\phi).\]  
This shows that the Schwarzschild spacetime is characterised by its invariants. Furthermore, the invariants are idential to the invariants given by the Euclidean version. Note that there is also a Lorentzian $(+---)$ version obtained by $\bar{\theta}=i\theta$ with identical invariants. This illustrates the subtleties in the difference of signatures and the invariants.   
  
In 5D the Schwarzschild spacetime is given by   
\beq  
\d s^2=-\left(1-\frac{2m}{r^2}\right)\d t^2+\frac{\d r^2}{\left(1-\frac{2m}{r^2}\right)}+r^2(\d \theta^2+\sin^2\theta \d \Omega^2).  
\eeq   
Here, there is a similar continuation, $\bar{t}=it$ that relates the $(-++++)$ version with the  
Euclidean $(+++++)$. Interestingly, in 5D we can also do the continuation $\bar{r}=ir$ which turns it into a Euclidean space of signature $(-----)$:  
\[  
\d s^2=-\left(1+\frac{2m}{\bar{r}^2}\right)\d t^2-\frac{\d \bar{r}^2}{\left(1+\frac{2m}{\bar{r}^2}\right)}-\bar{r}^2(\d \theta^2+\sin^2\theta \d \Omega^2).  
\]  
Note that this mapping does not leave a point invariant ($r=0$ is a singularity and therefore has to be omitted); therefore, the invariants are related via $\bar{I}(\bar{r})=I(r=-i\bar{r})$.  
\paragraph{A pair of  metrics:} Let us consider the pair of metrics:  
\beq\label{weakly1}  
\d s^2_1&=&\frac{1}{z^2}\left(\d x^2+\d y^2+\d z^2\right)-\d \tau^2, \\  
\d s^2_2&=&\frac{1}{z^2}\left(-\d x^2+\d y^2+\d z^2\right)+\d \tau^2. \label{weakly2}  
\eeq   
Now, we can easily see that these can be analytically continued into eachother. This shows that these metrics have indentical invariants. Furthermore, it is also clear that we can analytically continue them into a Euclidean metric of signature $(++++)$. Consequently, these metrics are characterised by their invariants.   
  
This pair of metrics illustrate two things. First,  
the metric (\ref{weakly1}) is of Segre type $\{1,(111)\}$ and is consequently (weakly) $\mathcal{I}$-non-degenerate\footnote{Recall that strongly $\mathcal{I}$-non-degenerate is defined to be the case where the inverse of the $\mathfrak{I}$-map consists of one point, and one point only ($\mathfrak{I}^{-1}(\mathcal{I}_0)\cong\redpoint$ in the notation in section \ref{sect:summary}), while weakly $\mathcal{I}$-non-degenerate is defined to be the case where the inverse of the $\mathfrak{I}$-map may consist of several isolated points (e.g., $\mathfrak{I}^{-1}(\mathcal{I}_0)\cong 3\redpoint$) \cite{inv}.}. The metric (\ref{weakly2}), on the other hand, is actually a Kundt degenerate metric. There are therefore continuous metric deformations of the latter metric to other metrics with identical invariants.   
  
Second, the metric (\ref{weakly1}) is weakly $\mathcal{I}$-non-degenerate but not strongly  $\mathcal{I}$-non-degenerate. The reason for this is the existence of an analytic continuation between two Lorentzian metrics. We might wonder if such spacetimes (which are weakly but not strongly  $\mathcal{I}$-non-degenerate) always have an associated analytic continuation.   
  
\section{Summary: the 4D Lorentzian case} 
\label{sect:summary}
Let us present a summary of the Lorentzian signature case. All of the results presented here are  
based on the 4-dimensional (4D) case, but it is believed that they are also true in higher-dimensions.   
  
Let us consider the space of 4D Lorentzian metrics $\mathfrak{M}$ over an open neighbourhood $U$.  
Consider also, $\mathcal{I}=\{I_i\}$, as the set of scalar polynomial invariants over $U$. The set of invariants  
can be considered as an element of the Cartesian product of smooth functions over $U$,  
$\left[C^{\infty}(U)\right]^{\times N}$ for some $N$. We can thus consider the 
calculation of invariants, $I_i$, from a metric $g_{\mu\nu}\in \mf{M}$ as a map 
$\mf{I}:\mf{M}\mapsto \left[C^{\infty}(U)\right]^{\times N}$, given by the $\mf{I}$-map:  
\[ \mf{I}: g_{\mu\nu}\mapsto I_i.\]   
This map is clearly not surjective, so for a given 
$\mathcal{I}\in \left[C^{\infty}(U)\right]^{\times N}$, the 
inverse image $\mf{I}^{-1}(\mathcal{I})$, may or may not be empty. Let us henceforth assume that we consider only points in the image, $\mathcal{B}\equiv\mf{I}(\mf{M})\subset \left[C^{\infty}(U)\right]^{\times N}$  of $\mf{I}$, and let us consider the sets of metrics with identical invariants; i.e., for a $\mathcal{I}_0\in \mathcal{B}$, $\mf{I}^{-1}(\mathcal{I})$.   
  
The connected components (in the set of metrics) of this inverse image can be of three kinds: 
\begin{enumerate}[i{)}]  
\item{} An isolated point. This would correspond to a metric being $\mathcal{I}$-non-degenerate. We will use the  symbol $\redpoint$ for such a point.   
\item{} A generic Kundt-wedge (or Kundt-tree; it is not entirely clear what the topology of this set is).  
This corresponds to a set of degenerate Kundt metrics \cite{Kundt}, none of which is  
characterised by its invariants, strongly or weakly. This set would be Kundt metrics connected via metric deformations  
with identical invariants. We will use the symbol $\tree$ to illustrate a set of this kind.  
\item{} A special Kundt-wedge containing a metric weakly characterised by its invariants. This set  
contains degenerate Kundt metrics: however, one of the members of this set is a special Kundt metric which is characterised by its invariants in the weak sense. We will use the symbol, $\redtree$, of a set of this kind. The symbol $\redpoint$ corresponds to the special point which is (weakly) characterised by its invariants, at the ``top'' of the Kundt-wedge.   
\end{enumerate}  
  
The various sets $\mf{I}^{-1}(\mathcal{I}_0)$, will consist of these components. In all of  
the examples known to the authors, the points $\redpoint$ are actually connected  
via analytic continuations. Figure \ref{fig1} presents a summary of the various cases that are known.  
\begin{figure}[tbp]  
\centering \includegraphics[width=10cm]{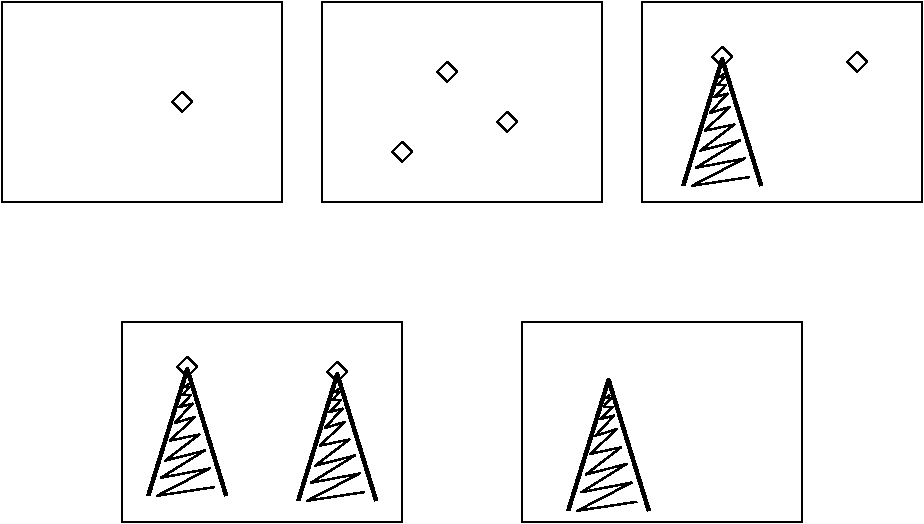}  
\caption{Figures showing the sets of metrics of identical invariants.}  
\label{fig1}  
\end{figure}  
  
Let us consider each figure in turn and comment on the various cases.  
\begin{enumerate}  
\item{} The generic case: Strongly $\mathcal{I}$-non-degenerate. This case consists of only one isolated point $\redpoint$. There is a unique metric with this set of invariants.   
\item{} Two or more isolated points $\redpoint$, all of which are (weakly) $\mathcal{I}$-non-degenerate.  
The Kasner metrics, eq. (\ref{Kasner}), are examples of this kind.   
\item{} Two or more components of types $\redpoint$ and $\redtree$. The pair of metrics, eqs. (\ref{weakly1}) and (\ref{weakly2}), are examples of this type and the metrics given corresponds to the two points $\redpoint$.   
\item{} One, two or more components of type $\redtree$. As an example of this kind are the two metrics $AdS_2\times S^2$ and $\mathbb{H}^2\times dS_2$.   
\item{} A generic Kundt-wedge, $\tree$. Examples of this case are generic degenerate Kundt metrics.  
\end{enumerate}  
It is not clear whether there might exist sets of type $\tree$ combined with any of the other  
components (like $\tree+\redpoint$) but no such examples are known to date.  
If $\mf{I}^{-1}(\mathcal{I}_0)$ consists of several components, then in all of the examples known  
to the authors, there is always an analytic metric continuation connecting points in them (not necessarily all points of the components).   
   
\section{Neutral Signature} 
Let us also briefly discuss the 4D pseudo-Riemannian  case of neutral 
signature $(--++)$ (NS space). Here, little has previously been done with regards to the connection between the invariants and the NS space. 
 
Let us first consider the Segre types of the Ricci tensor. This is essentially the canonical forms of of the Ricci operator and these will also be given for the non-standard types. Note that the canonical forms are obtained using real $SO(2,2)$ transformations. The possible types are \cite{Petrov}: 
\begin{itemize} 
\item{} $\{11,11\}$: All real eigenvalues. Degenerate cases are:  $\{(11),11\}$,  $\{11,(11)\}$,  $\{(11),(11)\}$,  $\{(1,1)(1,1)\}$,  $\{1(1,11)\}$, $(11,1)1\}$,  $\{(11,11)\}$. 
\item{} $\{211\}$. 
\item{} $\{31\}$. 
\item{} $\{4\}$: 4 equal eigenvalues. Canonical forms for Ricci operator, and metric, $g_{\mu\nu}$: 
\beq 
{\sf R}=\begin{bmatrix} 
\lambda & 1 & 0 & 0  \\ 
0 & \lambda &  1 & 0 \\ 
0 & 0 & \lambda &  1 \\ 
0 & 0 & 0 & \lambda  
\end{bmatrix} 
, \qquad 
(g_{\mu\nu})=  
\begin{bmatrix}  
0 & 0 & 0 & 1 \\ 
0 & 0 & 1 & 0 \\ 
0 & 1 & 0 & 0 \\ 
1 & 0 & 0 & 0 
\end{bmatrix}. 
\eeq 
\item{} $\{22\}$: Two non-diagonal Jordan blocks with 2 distinct eigenvalues. Canonical form: 
\[ {\sf R}=\diag({\sf B}_1,{\sf B}_2), \qquad (g_{\mu\nu})=\diag({\sf g}_1,{\sf g}_2),\] 
where  
\[ {\sf B}_A=\begin{bmatrix} 
\lambda_A & \pm 1 \\ 
0 & \lambda_A  
\end{bmatrix}, \qquad  
{\sf g}_A=\begin{bmatrix} 
0 & 1 \\ 1 & 0 
\end{bmatrix} \] 
\item{} $\{1_{\mathbb C}1,1\}$: 2 real and 2 complex conjugate eigenvalues. The block $1_{\mathbb C}$ has canonical form: 
\[ {\sf B}_A=\begin{bmatrix} 
\alpha_A & \beta_A \\ 
-\beta_A & \alpha_A  
\end{bmatrix}, \qquad  
{\sf g}_A=\begin{bmatrix} 
1 & 1 \\ 1 & -1 
\end{bmatrix} \] 
\item{} $\{1_{\mathbb C}1_{\mathbb C}\}$: 2 pairs of complex conjugate eigenvalues. Canonical form is  \[ {\sf R}=\diag({\sf B}_1,{\sf B}_2), \qquad (g_{\mu\nu})=\diag({\sf g}_1,{\sf g}_2),\] 
where  
\[ {\sf B}_A=\begin{bmatrix} 
\alpha_A & \beta_A \\ 
-\beta_A & \alpha_A  
\end{bmatrix}, \qquad  
{\sf g}_A=\begin{bmatrix} 
1 & 1 \\ 1 & -1 
\end{bmatrix} \] 
\item{} $\{2_{\mathbb{C}}\}$: 2 equal pairs of complex conjugate eigenvalues with a non-diagonal Jordan block. Canonical form: 
\[ {\sf R}=\begin{bmatrix} 
\alpha & \beta & 1 & 0 \\ 
-\beta & \alpha & 0 & 1\\ 
0 & 0 & \alpha & \beta \\ 
0 & 0 & -\beta & \alpha 
\end{bmatrix} 
, \qquad (g_{\mu\nu})=\begin{bmatrix} 
0 & 0 & 1 & 1 \\ 
0 & 0 & 1 & -1\\ 
1 & 1 & 0 & 0 \\ 
1 & -1 & 0 & 0 
\end{bmatrix}.\] 
\end{itemize} 
  
\begin{cor} 
Neutral signature metrics of Segre types  $\{11,11\}$, $\{(11),11\}$,  $\{11,(11)\}$,  $\{(11),(11)\}$, $\{1_{\mathbb C}1,1\}$ and  $\{1_{\mathbb C}1_{\mathbb C}\}$ are $\mathcal{I}$-non-degenerate, and are, consequently, characterised by the invariants in the strong sense.  
\end{cor} 
\begin{proof} 
The proof of this corollary is, except for types $\{(11),(11)\}$ and $\{1_{\mathbb C}1_{\mathbb C}\}$, identical to the Lorentzian case \cite{inv}. For type $\{(11),(11)\}$, we obtain two projectors, each with symmetry $SO(2)$. We can therefore project onto $SO(2)$-tensors and since this is a compact group, the invariants split orbits. Consequently, this is case is $\mathcal{I}$-non-degenerate. The case $\{1_{\mathbb C}1_{\mathbb C}\}$ is analogous to case $\{1_{\mathbb C}1,1\}$, but we need to perform a complex transformation for each of the two complex blocks. Since all eigenvalues must be different, we get that this case is also  $\mathcal{I}$-non-degenerate. 
\end{proof} 
 
For the Weyl operator ${\sf C}$, this splits into a self-dual and anti-self-dual part: ${\sf C}={\sf W}^+\oplus {\sf W}^-$. In \cite{Law} Law classified the Weyl tensor of NS metrics using the Weyl operator (or endomorphism). Using the Hodge star operator, $\star$, which commutes with the Weyl operator -- i.e., $\star \circ{\sf C}={\sf C}\circ\star$ -- the (anti-)self-dual operators can be defined as:  
\[ {\sf W}^\pm=\tfrac 12\left({\sf 1}\pm\star\right){\sf C}.\] 
 
Each of the parts can be considered to be symmetric and tracefree with respect to the 3-dimensional Lorentzian metric with signature $(+--)$. Consquently, each of the operators ${\sf W}^{\pm}$ can be classified according to ``Segre type''(the ``Type'' refers to Law's enumeration):  
\begin{itemize} 
\item{} Type Ia: $\{1,11\}$ 
\item{} Type Ib: $\{z\bar{z}1\}$ 
\item{} Type II: $\{21\}$ 
\item{} Type III:$\{3\}$.  
\end{itemize} 
It is also advantageous to refine Law's enumeration for the special cases:  
\begin{itemize} 
\item{} Type D: $\{(1,1)1\}$ 
\item{} Type N: $\{(21)\}$  
\end{itemize} 
 
Based on this classification, we get the following result: 
\begin{thm} 
If a neutral signature 4D metric has Weyl operators ${\sf W}^{\pm}$ both of type I, then the metric is $\mathcal{I}$-non-degenerate. 
\end{thm} 
\begin{proof} 
The proof utilises the fact that each type I operator ${\sf W}^{\pm}$, can pick up the $+$-direction of the fictitious 3D space with signature $(+--)$. In the real case, this gives rise to the eigenbivector $F_{\mu\nu}$ where, in the orthonormal frame,  
\[ 
{\sf F}^{\pm}\equiv (F^{\mu}_{~\nu})=\diag({\sf F}_1,\pm {\sf F}_1), \quad  
{\sf F}_1=\begin{bmatrix} 0 & 1 \\ 
-1 & 0 
\end{bmatrix} 
\]  
where the signs $\pm$ correspond to ${\sf W}^{\pm}$. Now by using these eigenvectors, we can construct the operator: 
\[ {\sf F}^+{\sf F}^-=\diag(-1,-1,1,1), \]  
which is of type $\{(11),(11)\}$. Consequently, using the argument from the Ricci type  $\{(11),(11)\}$, this NS space is $\mathcal{I}$-non-degenerate. 
\end{proof} 
Note that this is the generic case; that is,  
the general NS space is characterised by its invariants in the strong sense.  
 
Let us also discuss the NS case in light of theorem \ref{analytic}. As pointed out, 
the theorem is valid for any signature and thus applies to the NS case also. Consider the Euclidean 4D Schwarzschild spacetime as an example: 
\beq  
\d s^2=\left(1-\frac{2m}r\right)\d\tau^2+\frac{\d r^2}{\left(1-\frac{2m}r\right)}+r^2(\d \theta^2+\sin^2\theta \d \phi^2).  
\eeq   
Here, there are (at least) three complex metric extensions yielding an NS space, namely $(\bar{\tau},r,\theta,\bar{\phi})=(i{\tau},r,\theta,i{\phi})$,  $({\tau},r,\bar\theta,{\phi})=({\tau},r,i\theta,{\phi})$ and $(\bar{\tau},r,\bar\theta,\bar{\phi})=(i{\tau},r,i\theta,i{\phi})$, giving the triple of metrics: 
\beq 
\d s^2_1&=&-\left(1-\frac{2m}r\right)\d\bar{\tau}^2+\frac{\d r^2}{\left(1-\frac{2m}r\right)}+r^2(\d \theta^2-\sin^2\theta \d \bar{\phi}^2), \label{NS1} \\ 
\d s^2_2&=&\phantom{-}\left(1-\frac{2m}r\right)\d\tau^2+\frac{\d r^2}{\left(1-\frac{2m}r\right)}-r^2(\d \bar{\theta}^2+\sinh^2\bar\theta \d \phi^2), \label{NS2}\\
\d s^2_3&=&-\left(1-\frac{2m}r\right)\d\bar\tau^2+\frac{\d r^2}{\left(1-\frac{2m}r\right)}-r^2(\d \bar{\theta}^2-\sinh^2\bar\theta \d \bar\phi^2).  \label{NS3}
\eeq 
Therefore, these three metrics are characterised by its invariants. In order to determine whether they are strongly or weakly characterised by their invariants, a more thorough investigation of the NS case is needed. As a preliminary investigation we can try to determine the Weyl type of these metrics. By calculating ${\sf W}^\pm$ using the definitions, we get the following types:
\begin{enumerate}
\item{} Metric (\ref{NS1}): Type $\{(1,1)1\}\times \{(1,1)1\}$ or type D$\times$D. 
\item{} Metric (\ref{NS2}): Type $\{1,(11)\}\times \{1,(11)\}$ or type Ia$\times$Ia.
\item{} Metric (\ref{NS3}): Type $\{(1,1)1\}\times \{(1,1)1\}$ or type D$\times$D. 
\end{enumerate}
We can therefore conclude that metric (\ref{NS2}) is $\mathcal{I}$-non-degenerate, while the other metrics require a more thorough study.\footnote{The two metrics (\ref{NS1}) and (\ref{NS3}) are actually locally diffeomorphic, however, these two different forms are useful if we wish to extend these to degenerate metrics.}

In the Lorentzian case the spacetimes not characterised by their scalar invariants 
are Kundt spacetimes. We may therefore wonder what are the NS spaces that are not characterised by their invariants. A hint may be provided by utilising an analytic metric continuation of the following Lorentzian case.  
 
Consider the vacuum plane wave spacetimes \cite{exsol}: 
\beq 
\d s^2=2\d u(\d v+H(u,x,y)\d u)+\d x^2+\d y^2,  
\eeq 
where  $H(u,x,y)=f(u,x+iy)+\bar{f}(u,x-iy)$ for an analytic function $f(u,z)$. This is a well-known vacuum Kundt-VSI spacetime.  
Assume that we consider the special case $\bar{f}(u,\bar{z})=f(u,\bar{z})$. Then, we can perform the analytic metric continuation $\tilde{y}=iy$ which gives the NS metric: 
\beq 
\d s^2=2\d u(\d v+\tilde{H}(u,x,\tilde{y})\d u)+\d x^2-\d \tilde{y}^2,  
\eeq 
where $\tilde{H}(u,x,\tilde{y})=f(u,x+\tilde{y})+{f}(u,x-\tilde{y})$. 
Since the analytic continutation preserves the invariants and the structure, 
this metric must be VSI and a solution to the vacuum equations. Moreover, 
it cannot be characterised by its invariants and therefore represents  a ``Kundt analogue'' of the NS case. Interestingly, this metric is a special 4D Walker metric \cite{Walker}.  
 
The Walker metrics can be written as 
\beq\label{Walker}
\d s^2=2\d u(\d v+H_1\d u+W\d x)+2\d x(\d y+H_2\d x+W\d u),
\eeq
 where $H_1=H_1(u,v,x,y)$, $H_2=H_2(u,v,x,y)$ and $W=W(u,v,x,y)$. Such metrics admit a field of parallell null 2-planes. These metrics seem to possess some of the right curvature properties to be NS candidates for degenerate NS metrics \cite{curvW}. 

Let us introduce the double null-frame: 
\beq
{\mbold\omega}^1 =\d u && {\mbold\omega}^2=\d v+H_1\d u+W\d x\nonumber \\
{\mbold\omega}^3=\d x && {\mbold\omega}^4=\d y+H_2\d x+W\d u.
\label{doublenull}\eeq
For neutral metrics we can introduce two independent boost weights, $(b_1,b_2)$, corresponding to the two boosts: 
\beq
{\mbold\omega}^1 \mapsto e^{\lambda_1}{\mbold\omega}^1, && {\mbold\omega}^2 \mapsto e^{-\lambda_1}{\mbold\omega}^2,\nonumber \\
 {\mbold\omega}^3 \mapsto e^{\lambda_2}{\mbold\omega}^3, && {\mbold\omega}^4 \mapsto e^{-\lambda_2}{\mbold\omega}^4.
\eeq
Note that any invariant must have boost weight $(0,0)$, in addition, denoting $(T)_{(b_1,b_2)}$ as the projection of the tensor $T$ onto the boost weight $(b_1,b_2)$, then $(T)_{(b_1,b_2)}\otimes(\tilde{T})_{(\tilde{b}_1,\tilde{b}_2)}$ is of boost weight $(b_1+\tilde{b}_1,b_2+\tilde{b_2})$. Moreover, 
\[ (T\otimes\tilde{T})_{(\hat{b}_1,\hat{b}_2)}=\sum_{(b_1+\tilde{b}_1,b_2+\tilde{b}_2)=(\hat{b}_1,\hat{b}_2)}(T)_{(b_1,b_2)}\otimes(\tilde{T})_{(\tilde{b}_1,\tilde{b}_2)}.\]
Further, we will say that a tensor $T$, possesses the ${\bf N}$-property, iff: 
\beq
(T)_{(b_1,b_2)}=0,\quad\text{for}\quad
\begin{cases}
b_1>0, & b_2~\text{arbitrary},\\
b_1=0, & b_2>0, \\
b_1=0, & b_2=0.
\end{cases}
\eeq  
\begin{figure}[tbp]  
\centering \includegraphics[width=6cm]{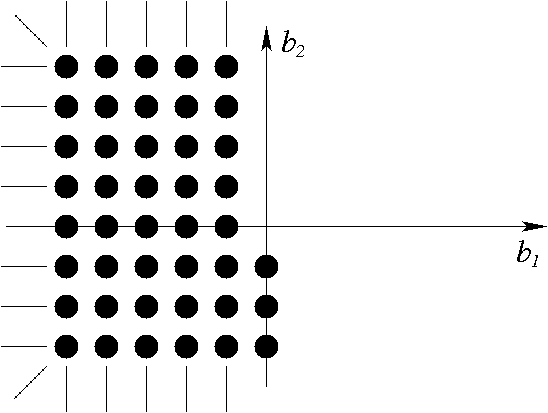}  
\caption{The tensors that fulfill the ${\bf N}$-property have components as indicated in the figure. The allowable boost weights, $(b_1,b_2)$, fill a semi-infinite grid in ${\mb Z}^2$ (which is closed under addition) indicated by black dots. }  
\label{fig2}  
\end{figure}  
Figure \ref{fig2} illustrates the allowable components in terms of their boost weight. 
We note that if both $T$ and $S$ possess the ${\bf N}$-property, then so does $T\otimes S$ since the resulting boost weights can be considered as ``vector addition'' in ${\mb Z}^2$. 

\begin{prop}
The eigenvalues of an operator, ${\sf T}$, possessing the  ${\bf N}$-property (as a tensor) are all zero; consequently, ${\sf T}$ is nilpotent.
\end{prop}
\begin{proof}
We note that the metric $g$ is of boost weight $(0,0)$; thus, if a tensor, ${\sf T}$ possesses the ${\bf N}$-property, so will any full contraction. Therefore, since an invariant must be of boost weight  $(0,0)$, $\Tr( {\sf T})=0$. Furthermore, since ${\sf T}^n$ must possess the ${\bf N}$-property also, $\Tr ({\sf T}^n)=0$. The eigenvalues are therefore all zero. 
\end{proof}

Let us now consider the Walker metrics. In the basis (\ref{doublenull}) we note that by counting the type of index (when all are downstairs!) we can get the boost weight as follows: $b_1=\#(1)-\#(2)$ and $b_2=\#(3)-\#(4)$. We are now in a position to state:
\begin{prop}
The Walker metrics eq.(\ref{Walker}) where
\beq
H_1&=& vH_1^{(1)}(u,x,y)+H_1^{(0)}(u,x,y)\nonumber \\
H_2&=& vH_2^{(10)}(u,x)+yH_2^{(01)}(u,x)+H_2^{(0)}(u,x)\nonumber \\
W&=& vW^{(1)}(u,x)+W^{(0)}(u,x,y),
\eeq
are NS metrics having the following properties: 
\begin{enumerate}
\item{} If $H_2^{(10)}(u,x)\neq 0$, then all 0th and 1st order invariants vanish; i.e., they are VSI$_1$ metrics.
\item{} If $H_2^{(10)}(u,x)=0$ then \emph{all} polynomial curvature invariants vanish; i.e., they are VSI metrics. 
\end{enumerate}
\end{prop}
\begin{proof}
Using the functions as defined in the theorem, then we get by direct calculation (using, for example, GRTensorII): 
\begin{itemize}
\item{} The Riemann tensor, $R$, possesses the ${\bf N}$-property. 
\item{} The covariant derivative of the Riemann tensor, $\nabla R$, possesses the ${\bf N}$-property.
\end{itemize}
Consequently, it is a VSI$_1$ space, and the first part of the theorem follows.

Concerning the last part, we note that for $H^{(10)}_2=0$, then the rotation coefficients, $\Gamma^\alpha_{~\mu\nu}$, possess the ${\bf N}$-property also. The covariant derivative of a tensor $T$ can symbolically be written 
\[ \nabla T=\partial T-\sum \Gamma*T,\]
Therefore, if both $T$ and $\Gamma$ possess the ${\bf N}$-property, then the only term that can prevent $\nabla T$ possessing the ${\bf N}$-property also is the partial derivative term $\partial T$.
We therefore need to check the terms that raise the boost weight in such a way that it could potentially violate the ${\bf N}$-property. 

For the Riemann tensor we can use the Bianchi identity, 
\beq
R_{\alpha\beta\mu\nu;\rho}=-R_{\alpha\beta\rho\mu;\nu}+R_{\alpha\beta\rho\nu;\mu},
\eeq
and for general tensors, the generalised Ricci identity,
\beq
[\nabla_{\mu},\nabla_{\nu}]T_{\alpha_1...\alpha_k}=\sum_{i=1}^kT_{\alpha_1...\lambda...\alpha_k}R^{\lambda}_{~\alpha_i\mu\nu}.
\eeq
We need to check the components $\nabla_{1}R_{\alpha\beta\mu\nu;\rho_1...\rho_k}$ and $\nabla_{3}R_{\alpha\beta\mu\nu;\rho_1...\rho_k}$ that can potentially violate the ${\bf N}$-property. We will do this by induction in $n$. From the first part of the theorem (which we have already shown), we know it is true for $n=0$ and $n=1$. Assume therefore its true for $n=k$ and $n=k-1$. Then we need to check for $n=k+1$. Let us first consider $\nabla_{1}R_{\alpha\beta\mu\nu;\rho_1...\rho_k}=R_{\alpha\beta\mu\nu;\rho_1...\rho_k 1}$.  We can now use the generalised Ricci identity (and possibly the Bianchi identity) to rewrite this component as a derivative w.r.t. $2$-, $3$-, or $4$-components. Of these, the only potentially dangerous term is the $R_{\alpha\beta\mu\nu;\rho_1...\rho_k 3}$, of boost weight $(0,0)$. If $\rho_k=1$ or $\rho_k=3$, then we use generalised Ricci to permute the indices so that $\rho_k$ is $2$ or $4$ (we can always do this). Then we use the generalised Ricci again to permute the last two indices $\rho_n$ and $3$. We now see that the $(0,0)$ component can be replaced with $\nabla_{2}R_{\alpha\beta\mu\nu;\rho_1...\rho_k}$ or $\nabla_{4}R_{\alpha\beta\mu\nu;\rho_1...\rho_k}$; consequently, these are zero. Hence, the tensor  $\nabla_{\rho_{k+1}}R_{\alpha\beta\mu\nu;\rho_1...\rho_k}$ fulfills the ${\bf N}$-property also. 

The theorem now follows from these results.
\end{proof}

This clearly indicates that the Walker metrics are indeed ``Kundt analogues'' for NS metrics, however, its not clear whether all degenerate NS metrics are Walker metrics. The investigation of the degenerate NS spaces will be left for future work. 
\section{Discussion}

In this paper we have further studied the question of when a pseudo-Riemannain manifold can be locally 
characterised by its scalar polynomial curvature invariants.
In particular, we have introduced the new concepts of diagonalisability
and analytic metric extension and proven some important theorems.
In the final two sections we have discussed some applications of these results. First,
we considered the 4d Lorentzian case, in part to illustrate the explicit theorems 
that are known  \cite{inv}. Second, we have initiated an investigation of the
neutral signature case, which has not been widely studied previously. We intend to continue the
study of the neutral signature case in future work.

\newpage

{\em Acknowledgements}. 
This work was in part supported by NSERC of Canada.


\begin{thebibliography}{abc}  
 
 
\bibitem{inv}  A. Coley, S. Hervik and N.   
Pelavas, 2009, Class. Quant. Grav. {\bf 26}, 025013   
[arXiv:0901.0791].   
   
   
   
\bibitem{Kundt} A. Coley, S. Hervik, G. Papadopoulos and N.   
Pelavas, 2009, Class. Quant. Grav. {\bf 26}, 105016   
[arXiv:0901.0394].    
   
\bibitem{bivect}  A. Coley and  S. Hervik,
2009, Class. Quant. Grav. {\bf 27},
015002  [arXiv:0909.1160].
 
   
\bibitem{class}  A. Coley, R. Milson, V. Pravda and A. Pravdova,   
2004, Class. Quant. Grav. {\bf 21}, L35 [gr-qc/0401008]; V. Pravda, A.   
Pravdov\' a, A. Coley and R. Milson, 2002, Class. Quant. Grav.   
{\bf 19}, 6213 [arXiv:0710.1598].  
A. Coley,  2008, Class. Quant. Grav.   
{\bf 25}, 033001 [arXiv:0710.1598].  

\bibitem{exsol} H. Stephani,D. Kramer,M. MacCallum and E. Herlt, 2003, 
\textit{Exact Solutions to Einstein's Field Equations}, 
(Second Edition, Cambridge Univ. Press). 


\bibitem{PTV} 
F. Pr\"ufer, F. Tricerri, L. Vanhecke, 1996, Trans. Am. Math. Soc. 
\textbf{348},  4643. 
 
\bibitem{Procesi} 
C. Procesi, 2007, \textit{Lie Groups -- an approach through Invariants and Representations}, Springer. 

\bibitem{shortinv}  A. Coley, S. Hervik and N.   
Pelavas, 2009, preprint.    
 
\bibitem{CSI}  A. Coley, S. Hervik and N.   
Pelavas, 2006, Class. Quant. Grav. {\bf 23}, 3053   
[arXiv:gr-qc/0509113] 
 
 
\bibitem{3CSI}  A. Coley, S. Hervik and N.   
Pelavas, 2008, Class. Quant. Grav. {\bf 25}, 025008   
[arXiv:0710.3903] 
 
\bibitem{4CSI}  A. Coley, S. Hervik and N.   
Pelavas, 2009, Class. Quant. Grav. {\bf 26}, 125011 
[arXiv:0904.4877]   

\bibitem{Kinnersley}  W. Kinnersley, 1969, J. Math. Phys. {\bf 10}, 1195.  
  

\bibitem{typeD} 
A. Coley and S. Hervik, 2009, 
Class. Quant. Grav. {\bf 26}, 247001 
[arXiv:0911.4923] 


%\bibitem{milson} A. Coley, R. Milson, G. Papadopoulos, preprint.  
  
\bibitem{Kasner} 
  E.~Kasner, 1921, 
  %``Geometrical theorems on Einstein's cosmological equations,'' 
  Am.\ J.\ Math.\  {\bf 43}, 217. 
 
  
\bibitem{Petrov} A. Z. Petrov, 1969, \textit{Einstein Spaces} (Pergamon Press). 
  
 
\bibitem{Law}  
P.R. Law, 1991, J. Math. Phys. \textbf{32}, 3039. 
 
\bibitem{Walker}  
A.G. Walker, 1950, Q. J. Math. Oxford \textbf{1} 69;  
A.G. Walker, 1950, Q. J. Math. Oxford \textbf{1} 147.  

\bibitem{curvW}
 M. Chaichi, E. Garcia-Rio and Y. Matsushita, 2005, Class. Quant. Grav. \textbf{22} 559.
\end{thebibliography}
\end{document}